\documentclass{article}

\usepackage{amsmath}
\usepackage{amsthm, amssymb}
\usepackage{color}
\usepackage{tikz, comment}
\usetikzlibrary{positioning, shapes}

\usepackage{times}
\usepackage{bm}
\usepackage{natbib}

\usepackage[plain,noend]{algorithm2e}

\definecolor{shadecolor}{RGB}{255,255,128}
\definecolor{shadecolor}{gray}{0.92}

\newtheorem{lemma}{Lemma}

\newcommand{\ubar}[1]{\text{\b{$#1$}}}

\begin{document}

\author{R. J. Meijer \and  T. J. P. Krebs \and  A. Solari \and J. J. Goeman}

\title{A shortcut for Hommel's procedure in linearithmic time}





\maketitle

\begin{abstract}
  Hommel's and Hochberg's procedures for familywise error control are both derived as shortcuts in a closed testing procedure with the Simes local test. Hommel's shortcut is exact but takes quadratic time in the number of hypotheses. Hochberg's shortcut takes only linearithmic time, but is conservative. In this paper we present an exact shortcut in linearithmic time, combining the strengths of both procedures. The novel shortcut also applies to a robust variant of Hommel's procedure that does not require the assumption of the Simes inequality.
\end{abstract}


\section{Introduction}

The method of \cite{Hommel1988} is a well-known multiple testing procedure that controls the familywise error rate (FWER), guaranteeing that with probability at least $1-\alpha$ no true null hypotheses are rejected. Hommel's method can be constructed using a combination of the closed testing procedure \citep{Marcus1976} with local tests based on the inequality of \cite{Simes1986}. Hommel's procedure is uniformly more powerful than the methods of Bonferroni \citep{Dunn1961}, \cite{Holm1979} and \cite{Hochberg1988}.

Hommel's procedure is valid whenever the Simes inequality can be assumed to hold for the $p$-values corresponding to true null hypotheses. The same condition is necessary for the validity of the procedure of \cite{Benjamini1995} as a False Discovery Rate (FDR) controlling procedure. This condition and has been extensively studied elsewhere \citep[e.g.][]{Benjamini2001, Rodland2006, Sarkar2008, Finner2014}. A robust variant of the procedure, which is more conservative but does not require the assumption of the Simes inequality, was proposed by \cite{Hommel1986}.

In general, calculation time of a procedure based on closed testing is exponential in the number of hypotheses $m$. Faster algorithms,  known as \emph{shortcuts}, can be constructed for specific choices of local tests. Shortcuts are called exact if they always yield the same number of rejections as the full closed testing procedure, or conservative of they yield at most the same number of rejections, and sometimes fewer. For the case of local tests based on Simes, \cite{Hommel1988} proposed an exact shortcut in quadratic time. \cite{Hochberg1988} presented an alternative shortcut that takes linear time after the $p$-values are sorted (which takes linearithmic time). Hochberg's shortcut, however, is conservative.

In this paper we present a novel shortcut for closed testing with local tests based on Simes. The new shortcut is exact like Hommel's procedure, but takes only linearithmic time like Hochberg's. It allows computationally efficient scaling of Simes-based multiple testing procedures to large multiple testing problems without the power loss incurred by switching to Hochberg's method. The new shortcut also generalizes to the robust variant of Hommel's method \citep{Hommel1986} that does not assume the Simes inequality.

The structure of the paper is as follows. We start by introducing closed testing and the Simes local test. From a result by Hommel, we introduce the crucial function $h(\alpha)$ and propose a novel algorithm to calculate it in linearithmic time. Next, we explain how to calculate all adjusted $p$-values from $h(\alpha)$ in linear time. We finish with a short simulation that compares runtimes in practice. The algorithms described in this paper are implemented in the \texttt{R}-package \texttt{hommel}, available on CRAN.

\section{Closed testing, Simes, Hommel and Hochberg}

Suppose we have $m$ hypotheses $H_1, \ldots, H_m$ that we are interested in testing. Some of the hypotheses are true, while others are false. Denote by $T \subseteq \{1,\ldots,m\}$ the index set of true hypotheses. For each hypothesis we suppose that we have corresponding (unadjusted) $p$-values $p_1,\ldots,p_m$, one for each hypothesis $H_1,\ldots H_m$. Without loss of generality we will assume that $p_1 \leq \ldots \leq p_m$.

The aim of a classical multiple testing procedure is to find an index set $R_\alpha$ of hypotheses to reject, which is as large as possible while still controlling FWER at level $\alpha$. This requires that
\[
\mathrm{P}(T \cap R_\alpha = \emptyset) \geq 1-\alpha,
\]
i.e.\ with probability at least $1-\alpha$ there are no type I errors. It is often convenient to describe $R_\alpha$ as a function of $\alpha$ through adjusted $p$-values $\tilde p_1,\ldots, \tilde p_m$, defined as follows:
\[
\tilde p_i = \min \{0\leq \alpha \leq 1\colon i \in R_\alpha\}.
\]
Generally, $R_\alpha$ is monotone in $\alpha$ with $R_1 = \{1,\ldots,m\}$ so that all $\tilde p_i$ are defined and we have $R_\alpha = \{1 \leq i\leq m \colon \tilde p_i \leq \alpha\}$.

Like Hommel and Hochberg we consider sets $R_\alpha$ arising from the combination of closed testing with Simes tests. The closed testing procedure \citep{Marcus1976} augments the collection of hypotheses with all possible intersection hypotheses $H_I=\bigcap_{i\in I} H_i$, with $I\subseteq\{1,\ldots,m\}$. An intersection hypothesis $H_I$ is true if and only if $H_i$ is true for all $i\in I$. Note that $H_i = H_{\{i\}}$, so that all original hypotheses, known as \emph{elementary hypotheses}, are also intersection hypotheses. Next, all intersection hypotheses are tested with a valid $\alpha$-level test, the \emph{local test}. The set $R_\alpha$ is now defined as the set of all $i \in \{1,\ldots,m\}$ for which $H_I$ is rejected by the local test for all $I \ni i$. For FWER to be controlled it is sufficient that the local test rejects $H_T$, the intersection of all true hypotheses, with probability at most $\alpha$ \citep{Goeman2011}.

The Simes local test rejects an intersection hypothesis $H_I$ if and only if there is at least one $i \in \{1,\ldots,|I|\}$ for which
\begin{equation} \label{eq_Simestest}
p_{(i\mathbin{:}I)} \leq \frac{i\alpha}{|I|},
\end{equation}
where for any $I \subseteq \{1,\ldots,m\}$ and $i \in \{1,\ldots,|I|\}$, we define $p_{(i\mathbin{:}I)}$ as the $i$th smallest $p$-value among the multiset $\{p_i\colon i\in I\}$.

For the validity of the resulting closed testing procedure we must assume that the Simes test rejects $H_T$ with probability at most $\alpha$. There is extensive and growing literature on the conditions under which this assumption holds, which we will not revisit here \citep{Benjamini2001, Rodland2006, Sarkar2008, Finner2014}. The validity of the Simes test on $H_T$ is also necessary for the validity of the procedure of \cite{Benjamini1995} as an FDR-controlling procedure.

A robust variant of the Simes test due to \cite{Hommel1983} rejects an intersection hypothesis $H_I$ if and only if there is at least one $i \in \{1,\ldots,|I|\}$ for which
\begin{equation} \label{eq_Simestest2}
p_{(i\mathbin{:}I)} \leq \frac{i\alpha}{|I|\sum_{j=1}^{|I|} j^{-1}}.
\end{equation}
This robust local test is more conservative but makes no assumption on the joint distribution of the $p$-values: it is valid whenever the distribution of $p$-values from true null hypotheses is uniform or stochastically larger than that. Its use as a local test in combination with closed testing was proposed by \cite{Hommel1986}. In our discussion below we will treat the tests (\ref{eq_Simestest}) and (\ref{eq_Simestest2}) as special cases of a general local test that rejects $H_I$ if and only if there is at least one $i \in \{1,\ldots,|I|\}$ for which
\begin{equation} \label{eq_Simestest_general}
s_{|I|} p_{(i\mathbin{:}I)} \leq i\alpha.
\end{equation}
The interesting cases will be $s_k = k$ (Simes) and $s_k = k\sum_{j=1}^{k} j^{-1}$ (robust variant), but in the rest of this paper we will only use that $s_0=0$ and that $s_k$ is weakly increasing in $k$. We also assume that all $s_k$, $k=1,\ldots, m$ can be calculated together in a most linearithmic time. Trivially, we reject all hypotheses when $\alpha=1$.

Naive application of the closed testing procedure requires $2^m$ local tests to be performed, which severely limits the usefulness of the procedure in large problems. For this reason \emph{shortcuts} have been developed for specific local tests. For the Simes local test
\cite{Hommel1988} proved that $i \in R_\alpha$ if and only if
\begin{equation} \label{hommel}
h(\alpha) p_i\leq\alpha,
\end{equation}
where $h(\alpha)$ is defined in Section \ref{section_h} below. Note that Hommel's rule is analogous to Bonferroni, except that the Bonferroni factor $m$ is replaced by the random variable $h(\alpha) \leq m$. Based on this result Hommel formulated a quadratic time algorithm for calculating all adjusted $p$-values.

To obtain an alternative shortcut for the Simes case, \cite{Hochberg1988} proved that $R'_\alpha \subseteq R_\alpha$, where $i \in R'_\alpha$ if and only if there is a $j \geq i$ such that
\[
(m-j+1) p_j \leq \alpha.
\]
This result leads to an easy linear time shortcut if the $p$-values are sorted, or a linearithmic time one if the sorting is taken into account. Hochberg's method is conservative, sacrificing power for computational efficiency. It is easy to find instances for which $R'_\alpha \subset R_\alpha$. For example when $m=4$, $\alpha=0.05$, $p_1 = 0.02$, $p_2 = 0.02$, $p_3 = 0.03$, and $p_4 = 0.90$, Hommel's method rejects two hypotheses, while Hochberg's method rejects none.

\section{The function $h(\alpha)$} \label{section_h}

To construct the new shortcut we start from Hommel's result (\ref{hommel}), generalizing it to the general local test (\ref{eq_Simestest_general}). The quantity in Hommel's procedure that is most costly to calculate is the function $h(\alpha)$. In this section we will study the properties of this function, before we show how it can be calculated for all $\alpha$ in linearithmic time.

We define $h(\alpha)$ as the largest size $|I|$ of an intersection hypothesis $H_I$ that cannot be rejected by the closed testing procedure at level $\alpha$. Equivalently, $h(\alpha)$ is therefore also the largest size $|I|$ of an intersection hypothesis $H_I$ that cannot be rejected by the local test at level $\alpha$.
If any hypothesis of size $i$ is not rejected, then certainly the `worst case' hypothesis $H_{K_i}$ with $|K_i|=i$ given by
\[
K_i = \{m-i+1, \ldots, m\}
\]
is not rejected. Applying (\ref{eq_Simestest_general}) it follows that
\[
h(\alpha) =\max \big\{i\in\{0,\ldots,m\}: s_ip_{m-i+j}>j\alpha, \textrm{ for } j=1,\ldots,i\big\}.
\]

The importance of $h(\alpha)$ is clear from the following lemma.

\begin{lemma} \label{elementary}
$H_i$ is rejected by the closed testing procedure if and only if $s_{h(\alpha)} p_i \leq \alpha$.
\end{lemma}
\begin{proof}
First, assume $s_{h(\alpha)} p_i \leq \alpha$ and take $I \ni i$. Now either $|I| > h(\alpha)$, so $H_I$ is rejected by the local test, or $|I| \leq h(\alpha)$, and
\[
s_{|I|} p_{(1\mathbin{:}I)} \leq s_{|I|} p_{i} \leq s_{h(\alpha)} p_i \leq \alpha,
\]
so $H_I$ is rejected by the local test. Consequently, $H_I$ is rejected by the local test for all $I \ni i$, so $H_i$ is rejected by the closed testing procedure.

Next, assume that $H_i$ is rejected by the closed testing procedure; we show that $s_{h(\alpha)} p_i \leq \alpha$. This is trivial if $h(\alpha)=0$, so assume $h(\alpha)>0$. Rejection of $H_i$ implies that $H_J$ is rejected for all $J \ni i$. Write $K = K_{h(\alpha)}$. Since $H_K$ is not rejected by definition of $h(\alpha)$ we must have $i \notin K$, so $i< m-h(\alpha)+1$. Consider $J = \{i, m-h(\alpha)+2, \ldots, m\}$. Since $H_K$ is not rejected, for all $2 \leq j \leq |J|$, we have $s_{|J|} p_{(j\mathbin{:}J)} = s_{|K|} p_{(j\mathbin{:}K)} > j\alpha$. Since $H_J$ was rejected by the local test we must, therefore, have $s_{h(\alpha)} p_i = s_{|J|} p_{(1\mathbin{:}J)} \leq \alpha$.
\end{proof}

The function $h(\alpha)$ is a right-continuous step function on its domain $[0,1]$, since it is weakly decreasing and takes integer values in $\{0,\ldots,m\}$. To describe $h(\alpha)$ it suffices, therefore, to find the jumps of the function. We have $h(1)=0$, and, if $s_1>0$, $h(0)$ is the number of non-zero $p$-values. Let
\[
\alpha_i = \min\{0\leq \alpha\leq 1\colon h(\alpha) < i\}.
\]
By definition of $h(\alpha)$, $\alpha_i$ for $1\leq i\leq m$ is the lowest $\alpha$-level at which the hypothesis $H_{K_j}$, the `worst case hypothesis of size $j$', can still be rejected by the local test procedure for all $j \geq i$. Therefore, for $i=1, \ldots m$, we have
\begin{equation} \label{cummax}
\alpha_i = \max_{i \leq j \leq m} \alpha^*_j,
\end{equation}
where $\alpha^*_j$ is the lowest $\alpha$-level at which the hypothesis $H_{K_j}$ can be rejected by the local test. By (\ref{eq_Simestest_general}) we have
\begin{equation}\label{alphamin}
\alpha^*_i = \min_{k\in\{1,\ldots,i\}} {\frac{s_i}{k} \cdot p_{m-i+k}} = s_i \cdot \min_{k\in\{1,\ldots,i\}} {\frac{p_{m-i+k}}{k}}.
\end{equation}
For $i>m$ we have $\alpha_i = \alpha^*_i=0$. Tied values of $\alpha_i$ indicate jumps of $h(\alpha)$ larger than 1. Since we may trivially reject all hypotheses if $\alpha=1$ we may set all values of $\alpha^*_i$ that exceed 1 to 1.

In the case of the Simes local test, things simplify slightly. It can be shown that $\alpha_i = \alpha^*_i$ for all $i$. Moreover, we can only have $\alpha_i=\alpha_j$ for $i\neq j$ if $\alpha_i=\alpha_j=0$ or if $\alpha_i=\alpha_j=\alpha_m=p_m$. This is intuitive from Figure \ref{fig_halpha}, where we illustrate the function $h(\alpha)$ and its jumps for the case of Simes local tests.

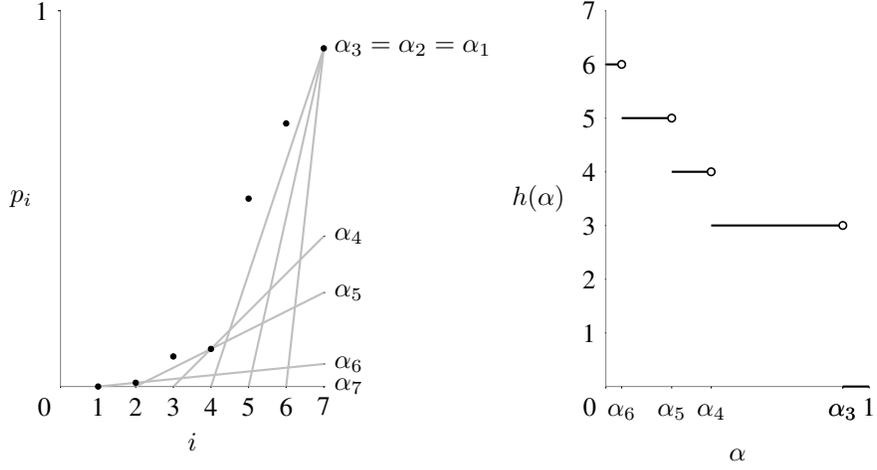
\begin{figure}[!ht]
\centering
\begin{tikzpicture}[scale=.5,
    mydot/.style={
      circle,
      fill=white,
      semithick,
      draw,
      outer sep=0pt,
      inner sep=1pt
    }]
    \draw[help lines] (0,0)--(7,0);
    \draw[help lines] (0,0)--(0,10) ;
    \draw (-2pt, 10) node[left] {1} -- (.2pt, 10);
    \draw[help lines,thick, gray!50] (1, 0) -- (7, 0.6);
    \draw[help lines,thick, gray!50] (2, 0) -- (7, 2.5);
    \draw[help lines, thick, gray!50] (3, 0) -- (7, 4);
    \draw[help lines, thick, gray!50] (4, 0) -- (7, 9);
    \draw[help lines, thick, gray!50] (5, 0) -- (7, 9);
    \draw[help lines, thick, gray!50] (6, 0) -- (7, 9);
    \draw[fill] (1,0) circle(2pt);
    \draw[fill] (2,0.1) circle(2pt);
    \draw[fill] (3,0.8) circle(2pt);
    \draw[fill] (4,1) circle(2pt);
    \draw[fill] (5,5) circle(2pt);
    \draw[fill] (6,7) circle(2pt);
    \draw[fill] (7,9) circle(2pt);
    \draw (7,0) -- (7cm+1pt,0) node[right]  {$\alpha_7$};
    \draw (7,0.6) -- (7cm+1pt,0.6) node[right]  {$\alpha_6$};
    \draw (7,2.5) -- (7cm+1pt,2.5) node[right]  {$\alpha_5$};
    \draw (7,4) -- (7cm+1pt,4) node[right]  {$\alpha_4$};
    \draw (7,9) -- (7cm+1pt,9) node[right]  {$\alpha_3=\alpha_2=\alpha_1$};
    \draw (0,-1pt) node[below left] {0} -- (0,0);
    \draw (1,-1pt) node[below] {1} -- (1,0);
    \draw (2,-1pt) node[below] {2} -- (2,0);
    \draw (3,-1pt) node[below] {3} -- (3,0);
    \draw (4,-1pt) node[below] {4} -- (4,0);
    \draw (5,-1pt) node[below] {5} -- (5,0);
    \draw (6,-1pt) node[below] {6} -- (6,0);
    \draw (7,-1pt) node[below] {7} -- (7,0);
    \path (3.5,-1) node[below]{$i$};
    \path (-.5,5) node[left]{$p_i$};
    \begin{scope}[shift={(14.5,0)}, xscale=.7, yscale=10/7]
    \draw[help lines] (0,0)--(10,0);
    \draw[help lines] (0,0)--(0,7);
    \draw (0,-1pt) node[below left] {0} -- (0,0);
    \foreach \y in {1,...,7}
      \draw (-1pt,\y) node[left] {\y} -- (0,\y);
    \draw[thick] (0,6) -- (0.6,6) node[mydot]{};
    \draw[thick] (0.6, 5) -- (2.5, 5) node[mydot]{};
    \draw[thick] (2.5,4) -- (4,4) node[mydot]{};
    \draw[thick] (4,3) -- (9,3)  node[mydot]{};
    \draw[thick] (9,0) -- (10,0);
    \draw (0.6,0) -- (0.6,-1pt) node[below]  {\phantom{1}$\alpha_6$\phantom{1}};
    \draw (2.5,0) -- (2.5,-1pt) node[below]  {\phantom{1}$\alpha_5$\phantom{1}};
    \draw (4,0) -- (4,-1pt) node[below]  {\phantom{1}$\alpha_4$\phantom{1}};
    \draw (9,0) -- (9,-1pt) node[below]  {\phantom{1}$\alpha_3$\phantom{1}};
    \draw (9,0) -- (9,-1pt) node[below]  {\phantom{1}$\alpha_3$\phantom{1}};
    \draw (10,0) -- (10,-1pt) node[below]  {1};
    \path (5,-1) node[below]{$\alpha$};
    \path (-1.2,3.5) node[left]{$h(\alpha)$};
    \end{scope}
\end{tikzpicture}
\caption{On the left-hand side an example of the construction of $\alpha_i$, ${i=1,\ldots,m}$, for the Simes local test with $m=7$ and $p_1=0$, $p_2=0.01$, $p_3=0.08$, $p_4=0.1$, $p_5=0.5$, $p_6=0.7$, $p_7=0.9$. The points are the $p$-values plotted against their rank. The value of $\alpha_i$ is the largest $\alpha$ such that the line from $(m-i, 0)$ to $(m,\alpha)$ has no points below the line. For the given $p$-values we find $\alpha_1=\alpha_2=\alpha_3=p_7=0.9$, $\alpha_4=4p_4=0.4$, $\alpha_5=5p_4/2=0.25$, $\alpha_6=6p_2=0.06$, $\alpha_7=7p_1=0$. The right-hand side is the resulting step function $h(\alpha)$.} \label{fig_halpha}
\end{figure}

\section{A linearithmic time algorithm for $h(\alpha)$}

Finding $h(\alpha)$ using (\ref{alphamin}) requires computing $m$ minima. These calculations are equivalent to finding the minimum of each column of the matrix
\begin{equation} \label{eq_M}
M = \left( \begin{array}{cccccc}
p_1 &  \vspace{.1cm} \\
\frac{p_2}{2} & p_2 & \vspace{.1cm}  \\
\frac{p_3}{3} & \frac{p_3}{2} & p_3 & \vspace{.1cm}  \\
\frac{p_4}{4} & \frac{p_4}{3} & \frac{p_4}{2} & p_4 & \vspace{.1cm}  \\
\vdots & \vdots & \vdots & \vdots & \ddots  \vspace{.1cm} \\
\frac{p_m}{m} & \frac{p_m}{m-1} & \frac{p_m}{m-2} & \frac{p_m}{m-3} & \dots & p_m
\end{array}\right),
\end{equation}
with $M_{ik} = p_i/(i-k+1)$, $i \geq k$. This calculation normally takes $\Theta(m^2)$ time. This is how Hommel's procedure is currently implemented, e.g.\ in the \texttt{R}-function \texttt{p.adjust}. However, the complexity of this calculation can be reduced to $O(m \log(m))$ by exploiting the following Lemma.

\begin{lemma}\label{lemma}
For $1\leq (k,l) \leq (i,j) \leq m$ with $(k-l)(i-j) \leq 0$, $M_{ik} \leq M_{jk}$ implies that $M_{il} \leq M_{jl}$.
\end{lemma}

\begin{proof}
Suppose $M_{ik} \leq M_{jk}$. Then $(j-k+1)p_i \leq (i-k+1)p_j$. Since \mbox{$(k-l)(i-j) \leq 0$} implies $(k-l)(p_i-p_j) \leq 0$ we may add $(k-l)p_i \leq (k-l)p_j$ to get $(j-l+1)p_i \leq (i-l+1)p_j$, so that $M_{il} \leq M_{jl}$.
\end{proof}

It follows from Lemma~\ref{lemma} that a sequence of minima of columns of $M$ can be found for which the row location is weakly increasing. If a minimum of column $k$ is in row $i$, a minimum of column $l<k$ will be in row $j\leq i$; a minimum in column $l>k$ will be in row $j\geq i$. This observation can be used to find a minimum in every column of $M$ in $O(m \log(m))$ steps by dividing the matrix in submatrices in every step and always computing a minimum in the middle column of a submatrix, as follows. We start by finding a minimum of the middle column $k$, which we will assume to be at row $i$. This will take $O(m)$ steps. The search for minima can now be naturally restricted to two submatrices; one consisting only of columns $1,\ldots,k-1$ and rows $1,\ldots, i$, and the second consisting only of columns $k+1,\ldots,m$ and rows $i,\ldots,m$. By Lemma~\ref{lemma} a minimum can always be found in these submatrices. Next, in both submatrices we find a minimum of their middle column. This time, calculating both minima simultaneously will take $m+1=O(m)$ steps. Again, each submatrix can be divided into two new submatrices using Lemma~\ref{lemma}, and we iterate. In each step, twice the number of minima compared to the previous step are calculated, and exhausting all $m$ columns will thus take $\left \lceil{\log_2(m)}\right \rceil $ iterations. Since in each step all minima can be simultaneously calculated in $O(m)$ time, the complexity of finding minima in all the columns of $M$ is $O(m\log(m))$. This is the complexity for finding $\alpha^*_1, \ldots, \alpha^*_m$.  The algorithm is summarized in Algorithm \ref{proc_M}, below.

\begin{algorithm}[!h]
\caption{Finding the jumps of $h(\alpha)$} \label{proc_M}
\vspace*{-12pt}
\begin{tabbing}
  \enspace Do $\mathrm{findjumps}(1, m, 1, m)$ where\\
  \enspace $\mathrm{findjumps}(\ubar{r}, \bar{r}, \ubar{c}, \bar{c})$:\\
  \qquad Set $j = \lfloor \frac12(\ubar{c} + \bar{c}) \rfloor$\\
  \qquad Find any $\max(\ubar{r}, j) \leq i \leq \bar{r}$, that minimizes $M_{ij} = \frac{p_i}{i-j+1}$\\
  \qquad Set $\alpha^*_{m-j+1} = s_{m-j+1} M_{ij}$\\
  \qquad If $j > \ubar{c}$ do $\mathrm{findjumps}(\ubar{r}, i, \ubar{c}, j-1)$\\
  \qquad If $j < \bar{c}$ do $\mathrm{findjumps}(i, \bar{r}, j+1, \bar{c})$\\
\end{tabbing}
\end{algorithm}

From $\alpha^*_1, \ldots, \alpha^*_m$ we can find $\alpha_1, \ldots, \alpha_m$ and therefore $h(\alpha)$ via (\ref{cummax}) in $O(m)$ time. Note that this final step may be omitted in the case of the Simes local test.

\section{Adjusted $p$-values}

From $h(\alpha)$ we can subsequently calculate the adjusted $p$-value $\tilde{p}_i$ for each hypothesis $H_i$ in Hommel's procedure by finding the minimum $\alpha$ for which $h(\alpha)p_i\leq\alpha$. We can, however, find all the adjusted $p$-values in linear time when we exploit the following lemma.

\begin{lemma} \label{lemma_adjusted}
The adjusted $p$-value of $H_i$ is given by
\[
\tilde{p}_i = \min (s_{t_i}p_i, \alpha_{t_i}),
\]
where
\begin{equation} \label{eq_adjp_ti}
t_i = \max \{1\leq j \leq m+1 \colon s_{j-1} p_i \leq \alpha_j\}.
\end{equation}
\end{lemma}

\begin{proof}
Let $\alpha \geq \tilde p_i$. Then $\alpha \geq \alpha_{t_i}$ or $\alpha \geq s_{t_i}p_i$. If $\alpha \geq \alpha_{t_i}$, then $h(\alpha) < t_i$. By definition of $t_i$ we have $s_{h(\alpha)} p_i \leq s_{t_i-1}p_i \leq \alpha_{t_i} \leq \alpha$, which implies by Lemma~\ref{elementary} that $H_i$ is rejected by the closed testing procedure. If $\alpha \geq s_{t_i} p_i$, then by definition of $t_i$ if $t_i\leq m$ and trivially if $t_i=m+1$, we have $\alpha \geq s_{t_i}p_i \geq \alpha_{t_i+1}$. Therefore $h(\alpha) \leq t_i$, so that $s_{h(\alpha)}p_i \leq s_{t_i}p_i \leq \alpha$, so $H_i$ is rejected by the closed testing procedure by Lemma~\ref{elementary}.

Let $\alpha < \tilde p_i$. Then $\alpha < s_{t_i}p_i$ and $\alpha < \alpha_{t_i}$. The latter implies that $h(\alpha) \geq t_i$, so that $s_{h(\alpha)} p_i \geq s_{t_i}p_i > \alpha$. This implies by Lemma~\ref{elementary} that $H_i$ is not rejected by the closed testing procedure.
\end{proof}

We can use the following procedure to find $t_1,\ldots, t_m$:

\begin{algorithm}[!h]
\caption{Calculating all adjusted $p$-values} \label{proc_adjp}
\vspace*{-12pt}
\begin{tabbing}
   \enspace Initialize $i=1$ and $j=m+1$\\
   \enspace Repeat\\
   \qquad If $s_{j-1}p_i \leq \alpha_j$ set $t_i = j$ and increment $i$ by 1\\
   \qquad Otherwise, decrease $j$ by 1 \\
   \qquad If $i>m$, stop\\
\end{tabbing}
\end{algorithm}

To see that the procedure is valid, note that for each $i$ the procedure finds the largest value of $1\leq j \leq t_{i-1}$ such that $s_{j-1}p_i \leq \alpha_j$. Since $p_i \geq p_{i-1}$ we have $t_i \leq t_{i-1}$, so this is equivalent to (\ref{eq_adjp_ti}). It is obvious that Algorithm \ref{proc_adjp} takes linear time in $m$, and that consequently all adjusted $p$-values can be calculated in linear time if $h(\alpha)$ is known.

\section{Implementation}

The algorithms in this paper have been implemented in \texttt{R} in the package \texttt{hommel} which is available on CRAN. We compare computation time in practice between the new algorithm and the current implementation of Hommel's and Hochberg's algorithms in \texttt{R}'s function \texttt{p.adjust}.

We sampled $m$ $p$-values independently from a squared uniform distribution, varying $m$ from $10^3$ to $5\times 10^7$. Computation times were evaluated from the \texttt{p.adjust}-function (package \texttt{base} for both the Hommel and Hochberg methods, and for the new algorithm as implemented in the \texttt{hommel} package. The results are given in figure \ref{comptimes}. Computation times for Hommel's method with \texttt{p.adjust} were not calculated for $m$ above $5\times 10^5$ since at this value of $m$ the calculation already took more than 5 hours. Computation times less than 10 minutes have been averaged over several runs.

\begin{figure}[!ht]
\includegraphics[width=\textwidth]{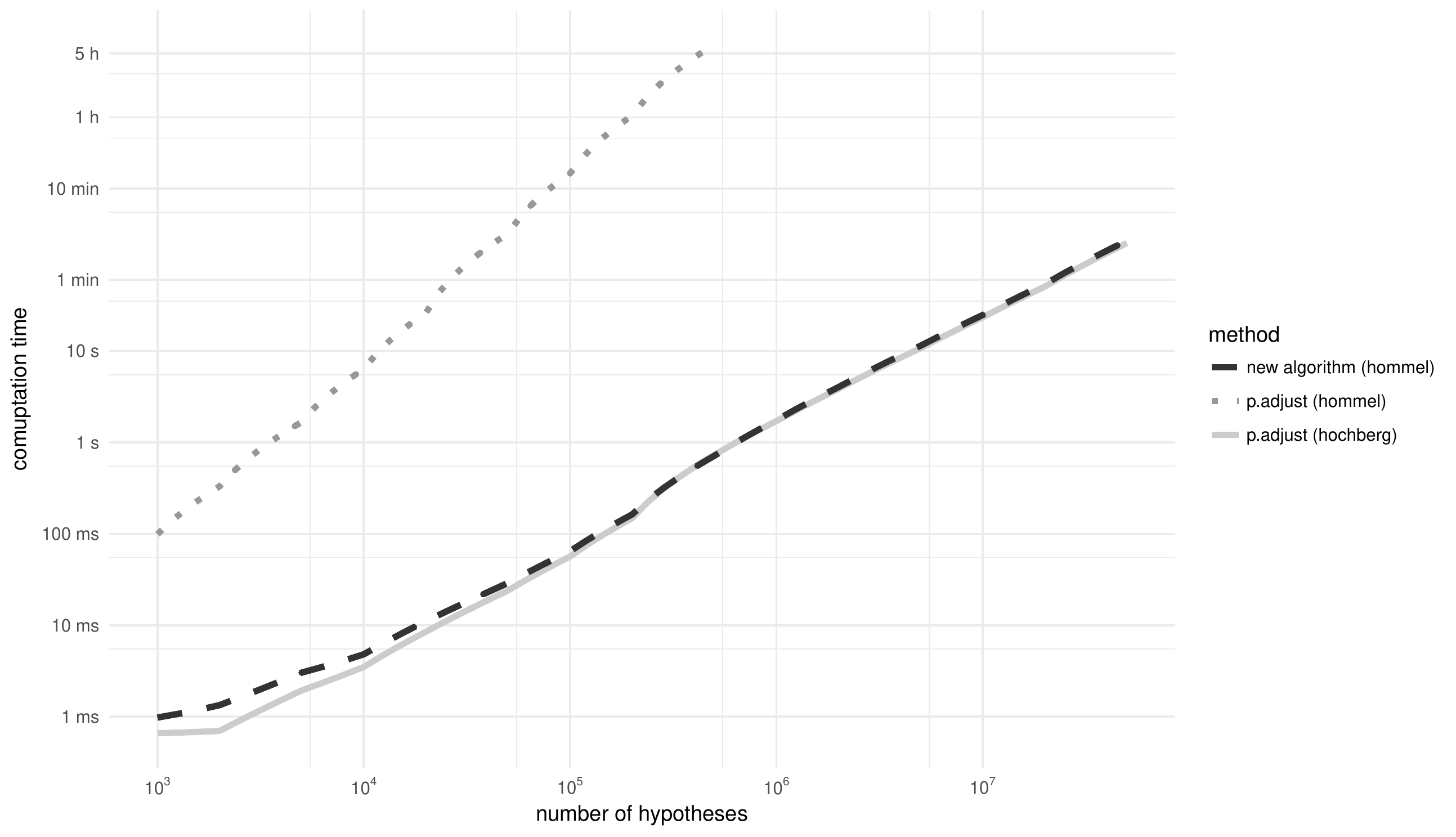}
\caption{Computation time for the new algorithm as implemented in the \texttt{hommel} package compared with Hommel's and Hochberg's methods as implemented in \texttt{p.adjust}.}
\label{comptimes}
\end{figure}

The computational gain of the new shortcut relative to Hommel's method is clear. We see that there is no apparent advantage for Hochberg's method over the new exact algorithm in terms of computation time.

\section{Discussion}

We have presented a new shortcut that calculates adjusted $p$-values of elementary hypotheses for closed testing with Simes local tests in linearithmic time. The new method combines the best of Hommel's and Hochberg's methods: it has the computational speed of Hochberg's method and the power of Hommel's method. The new shortcut should replace the two former methods in all applications.

The computational gain of the novel shortcut in this paper is important when many hypotheses are tested and FWER control is desired. Although FDR control is preferred in some application areas with many hypotheses (e.g.\ transcriptomics), FWER control remains the standard in other areas where more rigorous error control is desired, e.g.\ genome-wide association studies \citep{Sham2014} and neuroimaging \citep{Eklund2016}. In the end, the choice for FDR or FWER should not be made on the basis of the size of the multiple testing problem, but on the desired use and reliability of the discoveries. FWER control, unlike FDR control, for example, has the important property that error control remains guaranteed for arbitrary subsets of the discoveries \citep{Finner2001, Goeman2014}.

The combination of closed testing with the Simes local test can also be used for calculating confidence bounds for the false discovery proportion of subsets of the hypotheses \citep{Goeman2011}. The results of this paper may also be instrumental for obtaining such bounds more quickly.

\bibliographystyle{apalike}
\bibliography{hommel}

\end{document}